\documentclass[12pt]{iopart}
\usepackage{iopams}
\usepackage{graphicx}
\usepackage{bm}

\usepackage{theorem}
\newtheorem{theorem}{Theorem}
\makeatletter
\newenvironment{proof}[1][\proofname]{\par
  \normalfont
  \topsep6\p@\@plus6\p@ \trivlist
  \item[\hskip\labelsep{\bfseries #1}\@addpunct{.}]\ignorespaces
}{%
  \endtrivlist
}
\newcommand{\proofname}{\it{Proof}}
\makeatother

\def\qed{\hfill{$\square$}}

\begin{document}

\title[Typical Performance of Approximation Algorithms for NP-hard
Problems]{Typical Performance of Approximation Algorithms for NP-hard Problems}

\author{Satoshi Takabe$^1$ and Koji Hukushima$^{1,2}$}
\address{$^1$ Graduate School of Arts and Sciences, The University of Tokyo, 3-8-1 Komaba, Meguro-ku, Tokyo 153-8902, Japan}
\address{$^2$ Center for Materials Research by Information Integration,
National Institute for Materials Science, 1-2-1 Sengen, Tsukuba, Ibaraki 305-0047, Japan}
\ead{s\_takabe@huku.c.u-tokyo.ac.jp}

\begin{abstract}
 Typical performance of approximation algorithms is studied for
 randomized minimum vertex cover problems.  
 A wide class of 
 random graph ensembles 
 characterized by an arbitrary degree distribution is discussed with
 the presentation of a theoretical framework. 
 Herein, three approximation algorithms are examined: 
 linear-programming relaxation, loopy-belief propagation, 
 and a leaf-removal algorithm.
 The former two algorithms 
 are analyzed using a statistical--mechanical technique,
 whereas
 the average-case analysis of the last one is conducted using
 the generating function method. 
These algorithms have a threshold in the typical performance with
 increasing average degree of the random graph, below which they
 find true optimal solutions with high probability. 
 Our study 
 reveals that there exist only three cases 
determined by the order of the typical performance thresholds. 
 In addition, we provide 
 some conditions 
for classification of the graph ensembles 
 and demonstrate explicitly some examples
 for the difference in thresholds. 
\end{abstract}


\noindent{\it Keywords\/ average-case complexity, LP relaxation, cavity method, random graphs, scale-free networks}





\section{Introduction}
Evaluating the performance of approximation algorithms for optimization problems has attracted researchers' interests during the last several decades.
 It not only provides guarantees of approximations but also enables us to compare their performance for various cases.
 The approximation performance is roughly classifiable into worst-case performance and average-case performance.
 Although the former is an attractive research field in computer science~\cite{vazirani2013approximation},
  the latter is the issue that is mainly addressed in this paper.
 The average-case performance is based on behavior of the approximation algorithm averaged over randomized inputs of optimization problems.
 Probabilistic analyses of algorithms have been studied in the fields of computer science and probabilistic theory
  (e.g.~\cite{frieze1998probabilistic}).
 Especially for optimization problems defined on graphs, random inputs are regarded as random graphs.
 The average-case performance of graph algorithms is therefore examined on random graph ensembles~\cite{frieze1997algorithmic,coja2006max}.
 Typical behavior of the leaf-removal (LR) algorithm for a minimum vertex covers
 (min-VC), for instance, showing a phase transition with varying
 parameters characterizing the graph ensemble, 
  where its precision, the ratio of approximate values to true optimums, falls drastically ~\cite{bauer2001random}.
 
 Typical properties of optimization problems and their approximations
 are also attractive subjects in the spin-glass theory, which was
 originally developed to investigate spin-glass models with random
 interactions or fields~\cite{mezard1990spin}. 
 The spin-glass theory is then applied to study the average-case properties of optimization problems~\cite{fu1986application}, 
 revealing rich structures of optimization problems and their
 solutions~\cite{mezard2009information}. 
 The spin-glass theory also contributes to the development of an approximation algorithm
 called belief propagation (BP), 
 which enables the solution even of NP-hard optimization problems, intractable
 problems in their worst cases, 
 with high probability over random ensembles under a certain
 condition. 
 It is particularly interesting that the condition is described using a kind of 
phase transition that is related closely to so-called replica symmetry (RS)
 underlying the optimization problems and its breaking (RSB) transition
 called the RS--RSB transition.
 
Recently, the spin-glass theory and its techniques have also been applied to typical performance analyses of approximation algorithms,
  except for BP.
 The linear-programming (LP) relaxation is a widely used convex relaxation technique for combinatorial optimization.
 A statistical--mechanical analysis reveals that the typical behavior of LP relaxation for the min-VC
  on Erd\"os-R\'enyi random graphs shows a phase transition~\cite{Dewenter2012,takabe2014typical}.
 The transition of typical approximation performance occurs in the same
  condition as that of BP and LR, 
 suggesting that the RS--RSB transition in statistical physics is
  related to typical behavior not only of BP,
  but also of other approximation algorithms.
 Although previous works related to homogeneous random graphs support the
  phase-transition picture, it remains unclear whether it is the case for a
  much wider range of ensembles such as
 random scale-free networks, or not.
 As described herein, we study the typical behavior of those three approximations for min-VC on random graphs
 with arbitrary degree distribution.
 Along with some mathematically rigorous discussion, probabilistic and
  statistical--mechanical analyses 
  yield the condition of random graphs for which three approximations show the same typical performance.
 Moreover, we consider all possible cases of differences in their typical performance and provide examples for these scenarios. 
 
This paper is organized as follows. 
 In the following section, we define min-VC and its three
 approximation algorithms. 
 In \sref{sec_wor}, we compare their worst-case performance for the
 min-VC on an arbitrary graph, 
 which is useful to comprehend all possible cases in the average-case
 performance. 
 In \sref{sec_ran}, the typical performance of approximation algorithms is 
 defined using a concept of random graphs. We also provide a review of 
 some previous works examining
  typical behavior of approximations for the min-VCs on Er\"os-R\'enyi random graphs.
 It is the simplest case but it turned out to be fundamental in many aspects of the approximations, as shown in the next section.
 In \sref{sec_typ}, we present some theoretical typical analyses of these approximations. 
 These results indicate the existence of three cases related to the gap of their typical performance.
 In \sref{sec_exe}, some examples are provided to demonstrate these cases. They are studied using both theoretical and numerical analyses.
 The last section is devoted to a summary and discussion of the results.
 The Appendix presents a detailed RS cavity analysis of the
  statistical--mechanical model of LP relaxation. 
  is presented.

\section{Definition of vertex cover problem and its approximations}\label{sec_def}

Letting $G$ be an undirected graph without multi-edges and self-loops, we define $V$ and $E$ as a set of vertices and edges with respective cardinality of $N$ and $M$.
We cover vertices in $V$ to include at least one endpoint of each edge in $E$.
This problem is called the vertex cover problem.
Especially, the minimum vertex cover problem requests that one ascertains the minimum assignment of the covering.

The min-VC is represented by the integer programming (IP) problem.
Letting $x_i\in\{0,1\}$ be a variable on vertex $i\in V$ which takes $1$ if vertex $i$ is covered and $0$ otherwise, then, the problem reads as
\begin{eqnarray}
&\mbox{Min.\ \ \ \ \ \ \ \ \ } N^{-1}\sum_{i=1}^Nx_i,\nonumber\\
&\mbox{subject to \ } x_i+x_j\ge 1\quad& \forall (i,j)\in E, \label{eq_2}\\
&\mbox{\ \ \ \ \ \ \ \ \ \ \ \ \ \ \ } x_i\in\{0,1\} &\forall i\in V.\nonumber
\end{eqnarray}
The constraints on edges are represented by an incident matrix
$A_{\mathrm{inc}}=(a_{ai})\in \mathbb{Z}^{M\times N}$ of graph $G$, 
defined by $a_{ai}=1$ if edge $a$ connects to vertex $i$ and by $a_{ai}=0$ otherwise.
In \eref{eq_2}, the cost function is normalized by $N$ to take a $N\rightarrow\infty$ limit later.
We define its optimal value as $x_c^{\mathrm{IP}}(G)$ and simply call it the (true) optimal value of the problem on $G$.
Because the min-VC belongs to a class of NP-hard, it is usually difficult to
estimate the optimal value $x_c^{\mathrm{IP}}(G)$
 rigorously in polynomial time.
Alternatively, several polynomial-time approximation methods are available.
As described herein, we consider three algorithms based on different strategies.

The first one is called linear-programming relaxation.
This method solves a modified problem by changing a binary constraint $x_i\in\{0,1\}$ to a real one $x_i\in[0,1]$.
The relaxed problem therefore reads as
\begin{eqnarray}
&\mbox{Min.\ \ \ \ \ \ \ \ \ } N^{-1}\sum_{i=1}^Nx_i,\nonumber\\
&\mbox{subject to \ } x_i+x_j\ge 1\quad& \forall (i,j)\in E, \label{eq_3}\\
&\mbox{\ \ \ \ \ \ \ \ \ \ \ \ \ \ \ } x_i\in[0,1] &\forall i\in V.\nonumber
\end{eqnarray}
This change makes the problem tractable.
In general, this relaxation finds lower bounds of the original problem.
In LP relaxation, one can construct a simplex defined by linear constraints.
Because the cost function is also linear, there exist optimal extreme points on the simplex.
Extreme point solutions of LP are therefore important for the study of
LP structures.
The novel property of LP relaxation for min-VCs proven by 
Nemhauser and Trotter is half-integrality~\cite{nemhauser1974properties}:
 every extreme point for LP-relaxed min-VC is half-integral,
  that is, all elements are 0, 1/2, or 1.
From half-integrality, if an LP-relaxed solution has no half-integer elements, then LP relaxation finds a true optimal solution
of the min-VC.
It is also shown that an upper bound of the approximate value is a half.
These properties play a key role in later discussions.

The second method is the leaf-removal (LR) algorithm introduced by Karp and Sipser~\cite{karp1981maximum}.
This polynomial-time algorithm seeks a local optimum in a part of graph called a leaf.
A leaf is defined as a vertex with degree one.
It is locally optimal to cover $v$ instead of $w$ if vertex $v$ is connected to a leaf $w$.
The LR covers a root of a leaf and delete the covered vertex from the graph at each step.
This step is repeated until no leaves exist in the remains $G_\mathrm{R}$.
It covers all vertices in $G_\mathrm{c}$ to satisfy constraints if there exist connected components $G_\mathrm{c}$ when LR stops.
Results show that this approximation obtains an optimal value in the removed part of the graph.
Therefore, if a given graph is removed completely, then LR finds an optimal solution.
However, it usually fails to return the optimal value in the remainder $G_\mathrm{c}$ called an LR core.

The last method is loopy belief propagation (BP).
It is based on statistical--mechanical representations of optimization problems.
The min-VC on graph $G=(V,E)$ is represented by a hard-core lattice gas model with a partition function shown below.
\begin{equation}
Z=\sum_{\bm{x}}\exp\left(-\mu\sum_ix_i\right)\prod_{(i,j)\in E}\theta(x_i+x_j-1) ,\label{eq_21}
\end{equation}
In that equation, $\mu$ stands for a chemical potential of the system and $\theta(x)$ represents a step function which returns $1$ if $x\ge0$ and $0$ otherwise.

In BP, the 
marginal distribution of each variable is approximated by that of its nearest neighbors on a cavity graph $G\backslash i$.
It reads
\begin{equation}
P(x_i)\simeq Z_i^{-1}\sum_{\bm{x}_{\partial i}}e^{-\mu x_i}
\prod_{j\in \partial i}\theta(x_i+x_j-1)P_{j\rightarrow i}(x_j),\label{eq_22}\\
\end{equation}
where $\partial i = \{j\in V|\,(i,j)\in E\}$ and $Z_i^{-1}$ is a normalization constant.
Probability $P_{j\rightarrow i}(x)$ represents the 
marginal distribution of spin $x_j=x$ on a cavity graph $G\backslash i$.
On $G\backslash i$, the joint probability $P_{\partial i \rightarrow i}(\bm{x}_{\partial i})$
 of the nearest neighbors $\partial i$ 
is approximated by the product of $P_{j\rightarrow i}(x_j)$ ($j\in\partial i$).
This Bethe--Peierls approximation neglects correlation among neighboring spins. 

In this approximation scheme, these probabilities satisfy the following
recursive relations: 
\begin{equation}
P_{j\rightarrow i}(x_j)\simeq Z_{j\rightarrow i}^{-1}\sum_{\bm{x}_{\partial j\backslash i}}e^{-\mu x_j}
\prod_{k\in \partial j\backslash i}\theta(x_j+x_k-1)P_{k\rightarrow j}(x_k). \label{eq_23}
\end{equation}
Introducing a cavity field by $P_{j \rightarrow i}(x_j)=e^{-\mu h_{j\rightarrow i}x_j}/(1+e^{-\mu h_{j\rightarrow i}})$
and taking a large-$\mu$ limit,
we obtain a recursive relation of cavity fields as 
\begin{equation}
h_{j\rightarrow i}=1-\frac{1}{\mu}\sum_{k\in \partial j\backslash i}\ln(1+e^{\mu h_{k\rightarrow j}}) .\label{eq_25}
\end{equation}
A local field $h_i$ acting on a variable $x_i$ defined as $P(x_i)=e^{-\mu h_ix_i}/(1+e^{-\mu h_i})$ satisfies
\begin{equation}
h_{i}=1-\frac{1}{\mu}\sum_{j\in \partial i}\ln(1+e^{\mu h_{j\rightarrow i}}) .\label{eq_27}
\end{equation}
By solving these equations, BP provides an approximate value for a given graph.
Actually, BP is exact on a tree graph because of the lack of spin correlations~\cite{mezard2009information}.
In general, however, correlations among spins are affected by cycles on a graph.
Therefore, expecting the existence of fixed points, BP is used as an approximation.
The approximation use of BP equations \eref{eq_25} is called a loopy BP in the literature related to
 statistical physics~\cite{mezard2009information}.

\section{Relation of approximation algorithms for an arbitrary graph}\label{sec_wor}

Before defining the typical performance of approximations algorithms, we state some results related to their worst-case performance.
Because the worst-case results are available for arbitrary graphs,
  it is also useful to analyze the typical performance on an ensemble of random graphs.
As described below, LP invariably approximates min-VC better than
LR, although 
LP is not always superior to BP 
in general.

First, we state a theorem which claims that LP finds an optimal solution if LR finds it.
The proof is based on the strong duality theorem of the LP-relaxed problem~\cite{rockafellar1970convex} and the modified LR for dual problems.

\begin{theorem}\label{thm_1}
Letting $G=(V,E)$ be a graph that is removed completely by LR, then LP relaxation of the min-VC has an optimal solution for which an optimal value is equal to that
obtained using LR.
\end{theorem}

\begin{proof}
Letting $G_f=(V,F,E')$ be a factor graph representation of $G$, for which $F=E$, then the LP relaxed min-VC~\eref{eq_3} is represented as a standard form of
\begin{equation}
\mbox{Min. }x_c=N^{-1}\bm{c}^{\mathrm{T}}\bm{z},\:\mbox{s.t. }A\bm{z}= \bm{b},\, \bm{z}\ge \bm{0}, \label{eq_32}
\end{equation}
where $\bm{z}=(\{x_i\},\{z_a\})^{\mathrm{T}}\in\mathbb{R}^{N+M}$, 
$\bm{c}=(1,\cdots,1)^{\mathrm{T}}\in\mathbb{Z}^{N}$,
 $A=[A_{\mathrm{inc}},-I]\in\mathbb{Z}^{M\times (N+M)}$,  and 
$\bm{b}=(1,\cdots,1)^{\mathrm{T}}\in\mathbb{Z}^{M}$. 
Therein, $\{z_a\}$ are slack variables. $I$ represents an identity matrix of size $M$.
We also introduce $A_{\mathrm{inc}}\in\mathbb{Z}^{M\times N}$ as an incident matrix of graph $G$.
This primal LP problem~\eref{eq_32} has a dual problem given as
\begin{equation}
\mbox{Max. }y_c=N^{-1}\bm{b}^{\mathrm{T}}\bm{y},\:\mbox{s.t. }A^{\mathrm{T}}\bm{y}\le \bm{c}, \label{eq_33}
\end{equation}
where $\bm{y}\in\mathbb{R}^{M}$.
This optimization is equivalent to the following:
\begin{equation}
\mbox{Max. }y_c=N^{-1}\sum_{a}y_a,\:\mbox{s.t. }\sum_{a\in\partial i}y_a\le 1,\, y_a\ge 0. \label{eq_34}
\end{equation}

These primal and dual problems are feasible because $\bm{z}=(1,\cdots,1,0\cdots,0)$ and $\bm{y}=\bm{0}$ are, respectively, feasible solutions.
Considering that they are bounded, they have LP optimal values as $x_c^{\mathrm{LP}}=\min x_c$ and $y_c^{\mathrm{LP}}=\max y_c$.
The strong duality theorem then suggests that $x_c^{\mathrm{LP}}=y_c^{\mathrm{LP}}$ on every graph.

Next we return to the original min-VC and its dual one.
We respectively define the IP minimum and maximum values corresponding to \eref{eq_32} and \eref{eq_33} as $x_c^{\mathrm{IP}}$
and $y_c^{\mathrm{IP}}$.
Then, a trivial relation $y_c^{\mathrm{IP}}\le y_c^{\mathrm{LP}}=x_c^{\mathrm{LP}}\le x_c^{\mathrm{IP}}$ holds.
If LR can find $x_c^{\mathrm{IP}}$ and $y_c^{\mathrm{IP}}$ simultaneously and $x_c^{\mathrm{IP}}=y_c^{\mathrm{IP}}$ holds, then
  equalities in the above inequalities hold. An IP optimal solution $x_c^{\mathrm{IP}}$ obtained using LR is equal to $x_c^{\mathrm{LP}}$.

Next we demonstrate that LR finds a dual optimal solution with an optimal value equal to that of the primal min-VC.
At each step, LR searches leaf $L_a=\{\partial a|\,a\in F, \mathrm{deg}(i)= 1
(\exists i\in\partial a), \mathrm{deg}(j)\ge 1(\forall j\in\partial a\backslash i)\}$.
During its iterations, LR assigns a covered state to a variable node with maximum degree in $L_a$.
To solve the dual problem, in contrast, we need only to change
 an assigned node from the variable node
to functional node $a$.
When leaf $L_a$ is removed, its neighboring functional nodes are also removed.
This procedure is nothing but a witness of local optimality of the primal and dual IP problems.
Because LR assigns one variable or functional node in $L_a$ to a covered state,
  the 
 covered variables $x_c^{\mathrm{IP}}$ and $y_c^{\mathrm{IP}}$ are exactly equivalent.
We thus complete the proof. \qed
\end{proof}

However, a counterexample to the converse of the theorem exists.
On a bipartite graph without a leaf, LR cannot remove the graph and covers all vertices.
Using complete unimodularity of its incident matrix and the Hoffman--Kruskal theorem~\cite{hk},
it can be shown that LP relaxation returns an integral optimal solution.
These facts demonstrate that LP has greater ability to seek optimal solutions than LR does.
From the upper bound of LP relaxation, it is trivial that the minimum cover ratio does not exceed a half if LR can find it.

In the computer science literature, it has been revealed that an LP relaxed solution has a close relation
 to BP fixed points~\cite{sanghavi2007message}.
One can consider the minimum weighted vertex covers (min-WVC), min-VCs with a weighted cost function.
It can be shown that there exists an one-to-one map between approximate solutions obtained using the loopy BP and
extreme-point solutions of the LP-relaxed simplex.
The weights of the problems might be changed by this map, but the graph is invariant.
This fact indicates that LP and BP are not equivalent for a given min-WVC in general.
In fact, examples exist in which LP relaxation returns a true optimal solution but BP does not, and vice versa.
Although the difference in performance of LP and BP for related b-matching problems is shown~\cite{Bayati2011}, 
it is difficult to analyze them for min-VCs.

\section{Randomized min-VC and definition of typical performance}\label{sec_ran}
 The main purpose of this paper lies in evaluation of the average-case
 behavior of approximation algorithms, which requires 
 the setting of random graph ensembles 
 in contrast to the worst-case performance.
 As the approximation ratio in the worst-case analysis, a gap separating
 optimal and approximate values averaged over the random graph ensemble is a fundamental quantity to evaluate.
 As described in this paper, we specifically examine the simplest ensembles characterized by the degree distribution.
 
Let $\mathcal{G}(N)$ be a set of random graphs with a degree sequence having cardinality $N$ consisting of i.i.d. random variables
 of degree distribution $p_k$ ($k\ge 0$).
 For the analyses in this paper, we assume that $p_k$ is independent of $N$.
 The average degree $c$ is then defined as $c=\sum_kkp_k$.
 Then the weight of each graph in $\mathcal{G}(N)$ depends on its degree sequence.
 An optimal value of the min-VC averaged over random graphs $\mathcal{G}(N)$ in the thermodynamic limit is defined as
 \begin{equation}
 x_c^{\mathrm{IP}}(\mathcal{G})=\lim_{N\rightarrow\infty}\overline{x_c^{\mathrm{IP}}(G)}, \label{eq_31}
 \end{equation}
 where $\overline{(\cdots)}$ is an average over random graphs in $\mathcal{G}(N)$.
 Similarly, we define an average approximate value $x_c^{\mathrm{LP}}(\mathcal{G})$, $x_c^{\mathrm{LR}}(\mathcal{G})$, and 
 $x_c^{\mathrm{BP}}(\mathcal{G})$ respectively by LP relaxation, leaf removal, and BP.

 We present the typical behavior of approximations for the well-known Erd\"os-R\'enyi graphs as an example.
The Erd\"os-R\'enyi random graphs are generated by independently connecting an edge between each pair of vertices
 with probability $p=c/N(N-1)$.
 In the large-$N$ limit, its degree distribution converges to the Poisson distribution with mean $c$.
 The Erd\"os-R\'enyi random graph is therefore characterized by its average degree $c$.

In the statistical physics literature, the min-VC on the Erd\"os-R\'enyi random graph has been deeply studied~\cite{weigt2000number}.
 It is a notable result that mean-field theories such as the replica
 method and the RS cavity method succeed in estimating $x_c^{\mathrm{IP}}(c)$
 up to the so-called RS--RSB threshold $c=e(=2.71\cdots)$.
 Loopy BP and its variants succeed in estimating optimal values below the threshold~\cite{Weigt2006}.
 The convergence is also investigated numerically by evaluating spin-glass susceptibility~\cite{Zhang2009}. 
 Above the threshold, however, BP fails to converge because of strong correlations of neighboring spins.
 As for the typical performance of the loopy BP, there exists a phase transition at $c=e$.

Phase transition of the typical performance of LR is studied using a generating function~\cite{bauer2001random}, which revealed that a large LR core emerges and LR fails to
 approximate optimal values above the critical average degree $c=e$.

The LP relaxation is also investigated in terms of its typical behavior.
 Its transition has been reported numerically \cite{Dewenter2012} and analyzed theoretically in \cite{takabe2014typical}.
 As shown in the next section, the analysis is based on the cavity method for a three-state lattice--gas model called the LP--IP model.
 In the case of Erd\"os-R\'enyi random graphs, LP relaxation approximates min-VCs with high accuracy below the critical threshold $c=e$.
 Above the threshold, the minimum number of half-integers in LP-relaxed solutions is the order of $N$.
 Therefore, $x_c^{\mathrm{LP}}(c)<x_c^{\mathrm{IP}}(c)$ indicating incorrect approximation by LP relaxation.
 
 These studies reveal that three approximation methods for
 min-VCs have a phase transition of typical performance 
 at the same critical threshold at which the RS--RSB transition occurs.
 The motivation of this paper is to ascertain whether this relation holds in other ensembles, or not.
 If not, differences exist in typical performance of
 approximations because the worst-case performance depends significantly on
 approximation algorithms.

\section{Typical performance analyses of approximations}\label{sec_typ}
This section presents a description of theoretical analyses of typical behavior of approximations.
We basically use tree approximations, a mean-field approximation for graphs.
Because it is difficult to analyze LP relaxation directly, we apply a statistical--mechanical technique
 to an effective three-state lattice--gas model.
Theoretical analyses enable us to predict the threshold of typical behavior of three approximation algorithms.
In the following subsection, we discuss possible magnitude relations of their thresholds.

\subsection{LR}
Typical properties of LR are derived from theoretical analyses based on the generating function method.
Considering a rooted tree, one can evaluate its typical behavior including an average approximate value and the LR core fraction.
Some works have examined LR for the min-VC \cite{bauer2001random,takabe2014minimum}
and related problems \cite{zhao2015statistical}.
Recently, general results for the min-VC on random graphs with an arbitrary degree distribution are shown~\cite{Lucibello2014}.

In the generating function method, it is assumed that LR can remove at least a leaf.
To consider all possible cases, we show a theorem extended from the original one in \cite{Lucibello2014}.
We restrict the original theorem to the min-VC and add it to the case
in which no leaves are shown on the graph.

\begin{theorem}{}\label{thm_2}
Assume that a given graph ensemble is characterized by degree distribution $p_k$.
If LR cannot work at all, i.e., $p_1=0$, then $x_c^{\mathrm{LR}}(c)=1-p_0$.
Otherwise, let $g(x)$ be a continuous and increasing function represented as
\begin{equation}
g(x)=\sum_k\frac{kp_k}{c}x^{k-1}. \label{eq_52}
\end{equation}
Given that $X$ and $Y(\le X)$ satisfy relations $X=g(1-Y)$, $Y=g(1-X)$ and $0\le Y\le X \le1$,
the average approximate value obtained using LR is represented as
\begin{equation}
x_c^{\mathrm{LR}}(c)=1-\frac{c}{2}(X^2+2XY-2Y^2)+\sum_kp_k\{(1-Y)^k-2(1-X)^k\}. \label{eq_53}
\end{equation}
Especially, LR typically works well without generating a large LR core iff $X=Y$.
\end{theorem}

\begin{proof}
For the case in which $p_1=0$, LR finds no leaves.
The LR core then consists of connected components in a given graph, for which the average cardinality is equal to $(1-p_0)N$.
The average cover ratio is therefore equal to $1-p_0$ because all vertices are covered in the LR core.

Otherwise, LR removes $O(N)$ vertices with high probability.
Its performance is evaluated using the probabilistic analysis of the rooted tree.
Details are omitted here because one must only slightly modify the proof in~\cite{karp1981maximum}.
\qed
\end{proof}

For almost all cases in which $p_1>0$, the average fraction of the LR core $r_{\mathrm{c}}$ is represented as
\begin{equation}
r_{\mathrm{c}}(c)=\sum_kp_k\{(1-Y)^k-(1-X)^k\}-c(X-Y)Y. \label{eq_54}
\end{equation}
Then the condition $X=Y(<1)$ is equivalent to almost complete deletion of graphs.
As described in \sref{sec_def}, this means that LR approximates optimal values typically with high accuracy.
Otherwise, the emergence of a large LR core results in incorrect estimation of the algorithm.
The threshold $c^{\mathrm{LR}}$ above which LR typically fails good approximation of the problem
is given by the linear instability of the fixed point satisfying $X=Y$.
We therefore obtain the condition as $|g'(1-X(c^{\mathrm{LR}}))|=1$.

\subsection{LP and BP}\label{sec_lpip}
Typical behavior of LP and BP for min-VCs is analyzed using the same model with different parameters.
Focusing on the half-integrality property of LP relaxation, the LP--IP model for min-VCs is represented by the three-state Ising model
with a Hamiltonian represented by
\begin{equation}
\mathcal{H}_r(\bm{x})=-\sum_{i=1}^N x_i+\mu^{r-1}\sum_{i}\delta_{x_i,1/2}, \label{eq_m1}
\end{equation}
where $x_i\in\{0,1/2,1\}$ and $r$ is a constant parameter for a penalty term.
With appropriate parameter $r$ fixed, the LP--IP model in the large-$\mu$ limit describes optimal solutions obtained using LP and IP.
For the case in which $r>1$, the penalty terms prohibit each spin taking a half.
Consequently, ground states consist of integers resulting in the IP optimal solution.
We designate this limit as an IP-limit.
When $0<r<1$, ground-state energy is equivalent to the LP-relaxed value
assuming the half-integrality. 
The ground states include the minimum number of half integers.
This limit is defined by an LP-limit.
We are therefore able to analyze the typical behavior of LP and BP
through mean-field analysis of the LP--IP model.
Details of the analysis using the RS cavity method are described in the Appendix.

The averaged optimal value of the min-VCs is given as
\begin{equation}
 x_c^{\mathrm{IP}}(c)=x_c^{\mathrm{BP}}(c)=1-\frac{c}{2}X^2-\sum_kp_k(1-X)^{k}, \label{eq_m3}
\end{equation} 
where $X$ satisfies an equation $X=g(1-X)$.
This result is based on the RS ansatz, which corresponds
to the typical analysis of the loopy BP.
It works well, supported by numerical simulations, if the average degree is small.
In contrast, above some threshold, the RS solution is unstable against perturbation.
It is known as the RS--RSB threshold in the spin-glass theory.
To examine the stability, the de Almeida--Thouless
condition~\cite{de1978stability} is often used, but is difficult to apply to our case.
We therefore use an alternative linear stability of RS solutions~\cite{zdeborova2006number}. 
Then, the threshold of average degree $c^{\mathrm{IP}}$ satisfies $|g'(1-X(c^{\mathrm{IP}}))|=1$.
Below $c^{\mathrm{IP}}$, the RS solution is linearly stable.
In terms of the loopy BP, $c^{\mathrm{IP}}$ is regarded as the threshold for which the fixed point predicted under the RS ansatz exists.
We therefore naively assume that $c^{\mathrm{BP}}=c^{\mathrm{IP}}$ and use $c^{\mathrm{IP}}$ as the performance threshold of BP.

However, the LP-relaxed approximate value averaged over the same random graphs is
\begin{equation}
 x_c^{\mathrm{LP}}(c)=1-\frac{c}{2}XY-\frac{1}{2}\sum_kp_k\left\{(1-Y)^k+(1-X)^{k}\right\}, \label{eq_m4}
\end{equation}
 and an average ratio of the LP core, vertices on which variables take a
 half, follows
\begin{eqnarray}
 p_h(c)&=\sum_{k}p_k\left\{(1-Y)^k-(1-X)^k\right\}-cY(X-Y) \label{eq_m5}
\end{eqnarray}
where $X$ and $Y(\le X)$ satisfy the following equations,
\begin{equation}
 X=g(1-Y),\quad Y=g(1-X). \label{eq_m6}
\end{equation}

It is apparent that $x_c^{\mathrm{LP}}(c)=x_c^{\mathrm{IP}}(c)$ and $p_h(c)=0$ under the condition $X=Y$, which implies that LP relaxation works typically with good accuracy.
Otherwise, one finds that $x_c^{\mathrm{LP}}(c)<x_c^{\mathrm{IP}}(c)$ and $p_h(c)>0$, which suggests that the LP relaxed solution
 usually consists of numerous vertices with the half-integers. 
It is worth noting that equations \eref{eq_m6} for $X$ and $Y$ of LP relaxation correspond to those of LR,
although $x_c^{\mathrm{LP}}(c)\neq{x_c^{\mathrm{LR}}}(c)$ in general.
If $X=Y$, however, then $x_c^{\mathrm{LP}}(c)={x_c^{\mathrm{LR}}}(c)$ holds
resulting in $x_c^{\mathrm{IP}}(c)=x_c^{\mathrm{LP}}(c)={x_c^{\mathrm{LR}}}(c)$.
This fact shows that unless LR cannot work at all,
it approximates the problem well as long as LP relaxation does.
Considering the linear stability denoted above, the solution such that $X=Y$ is stable under some average degree $c^{\mathrm{LP}}$
given by $|g'(1-X(c^{\mathrm{LP}}))|=1$.
This condition is nothing but those for $c^{\mathrm{IP}}$ and $c^{\mathrm{LR}}$ under $p_1>0$.
We therefore conclude that, if $p_1>0$, then
  $c^{\mathrm{LR}}=c^{\mathrm{LP}}=c^{\mathrm{IP}}$ holds as the case of Erd\"os-R\'enyi random graphs.
In terms of typical performance of approximations, it claims that the three methods, BP, LP, and LR fail good approximations at the same
threshold for random graphs with numerous leaves.
Additionally, it is valid that the average fraction of the LR core $r_c(c)$ is equivalent
 to that of the LP core $p_h(c)$ in LP-relaxed solutions.

Another stable solution is $(X,Y)=(1,0)$, giving the upper bound of the
LP relaxation $x_c(c)=1/2$. 
The solution $X=Y$ usually engenders $x_c^{\mathrm{LP}}(c)\le 1/2$; it represents the ground-state energy of the system.
If the solution $X=Y$ provides $x_c^{\mathrm{LP}}(c)> 1/2$, then the
ground state is given by the solution 
$(X,Y)=(1,0)$. 
Its condition is then represented by $x_c^{\mathrm{IP}}(c)> 1/2$
 because $x_c^{\mathrm{LP}}(c)$ with $X=Y$ is equivalent to $x_c^{\mathrm{IP}}(c)$.

\subsection{Difference of typical performance: possible scenarios}\label{sec_dif}
Here we discuss possible cases of difference in typical performance
derived from the analytical results.
From Theorem~\ref{thm_1}, LP works better than LR even in terms of typical performance.
The theorem indicates that $c^{\mathrm{LR}}\le c^{\mathrm{LP}}$ always holds; the inequality is strict iff $p_1=0$.
The RS analyses of LP and BP engender $c^{\mathrm{LP}}\le c^{\mathrm{IP}}$,
  where it is strict if $x_c^{\mathrm{IP}}(c^{\mathrm{IP}})> 1/2$.

From these results, we find four cases:
(i) $c^{\mathrm{LR}}=c^{\mathrm{LP}}=c^{\mathrm{IP}}$; (ii) $c^{\mathrm{LR}}<c^{\mathrm{LP}}=c^{\mathrm{IP}}$;
(iii) $c^{\mathrm{LR}}<c^{\mathrm{LP}}<c^{\mathrm{IP}}$; and (iv) $c^{\mathrm{LR}}=c^{\mathrm{LP}}<c^{\mathrm{IP}}$.
However, case (iv) is denied because of the following reasons:
if case (iv) is true, then $p_1>0$ and $x_c^{\mathrm{LP}}(c^{\mathrm{LP}})=1/2$ holds.
Between two thresholds $c^{\mathrm{LP}}$ and $c^{\mathrm{IP}}$, the
solutions of equation \eref{eq_m6} are $(X,Y)=(1,0)$ and $(x,x)$ ($0<x<1$).
If $(X,Y)=(1,0)$ is applied, then LR generates almost all parts of a graph as a LR core,
  but this yields a contradiction because $O(N)$ vertices are always removed by $p_1>0$.
Otherwise, LR typically finds true optimal values, but it conflicts with the fact that $x_c^{\mathrm{IP}}(c)>1/2$ if $c>c^{\mathrm{LP}}$ and
that LR returns optimal values below a half.

In summary, these theoretical analyses reveal that three cases exist on the threshold of typical performance:
(i) $c^{\mathrm{LR}}=c^{\mathrm{LP}}=c^{\mathrm{IP}}$; (ii) $c^{\mathrm{LR}}<c^{\mathrm{LP}}=c^{\mathrm{IP}}$;
and (iii) $c^{\mathrm{LR}}<c^{\mathrm{LP}}<c^{\mathrm{IP}}$.
Especially, if $p_1>0$, i.e., then LR removes $O(N)$ vertices. Three algorithms work well on the same graph ensemble.
In the next section, we present some examples satisfying cases (ii) and (iii).

\section{Examples: failure of LR and LP}\label{sec_exe}
In the last section, we find theoretically that BP, LP, and LR usually have the same ability of typical approximation for the min-VC.
We also predict that, in some cases, these methods will have different typical performance, as is true with their worst performance shown in \sref{sec_wor}.
In this section, we provide some examples in which BP and LP have better typical performance than LR.
We also describe a special case in which BP has the best typical performance among them.

\subsection{Regular random graphs and their variants}
The simplest case in which LR cannot work at all is $K(\ge 2)$-regular random graphs.
A statistical--mechanical analysis of the min-VC on regular random graphs reveals that the RS solution is stable iff $K=1,2$~\cite{Barbier2013}.
In contrast, apparently, LR cannot work at all on the $2$-regular random graphs.
This is a trivial example in which the only LR typically fails to approximate the problem with high accuracy, i.e., 
$c^{\mathrm{LR}}=1<c^{\mathrm{LP}}=c^{\mathrm{IP}}=2$.

Considering graph ensembles with fluctuating vertex degree, it is
nontrivial whether there is an ensemble yielding 
$c^{\mathrm{LR}}<c^{\mathrm{LP}}=c^{\mathrm{IP}}$.
Random graphs with a bimodal degree distribution are a natural extension of the regular random graphs.
For example, the bimodal degree distribution is represented as $p_k=(1-p)\delta(k-K)+p\delta(k-K-1)$ with average degree $K+p$ ($0\le p <1$).
However, we find that
$c^{\mathrm{LR}}=2-\epsilon<c^{\mathrm{LP}}=c^{\mathrm{IR}}=2$ ($\forall
\epsilon >0$), suggesting that
differences of typical performance emerge only in the $2$-regular graphs.
From these observations, we cannot ascertain whether the gap of typical performance results from a homogeneous property
 of the regular random graphs, or not.
In some examples considered later, we consider inhomogeneous random graphs in which a key of difference is
not homogeneity but their degree bounds.

\subsection{BA-like scale-free networks}
To examine inhomogeneous random graphs with a continuous average degree, we define random graphs with the degree distribution given as
\begin{equation}
p_k=
\cases{
\frac{2(1-p)}{m+2} &($k=m$),\\
\frac{2m(m+1)}{k(k+1)(k+2)}(1-p)+\frac{2(m+1)(m+2)}{k(k+1)(k+2)}p &($k>m$), 
}
\label{eq_ba1}
\end{equation}
where $0\le p< 1$. 
Degree distribution $p_k$ is a mixture of two degree distributions appearing
 in the Barabasi--Albert (BA) model~\cite{barabasi1999emergence,dorogovtsev2000structure}.
By introducing $p$ as a parameter, its average degree $c=2(m+p)$ is a real number, although the original BA model has a discrete average degree.
The distribution is truncated with $m$ representing 
 the lower degree bound of graphs.
Although the original BA model is generated dynamically, we statically construct a random graph as a configuration model~\cite{bender1978asymptotic}, which enables us to avoid some intrinsic correlations in the BA model.

If we set $m\ge 2$, LR cannot work at all, i.e., $c^{\mathrm{LR}}=4-\epsilon$ ($\forall \epsilon>0$).
It is a main subject whether LP and BP are still good approximations in this case.
By solving equations in the last section, their theoretical estimations are available.
\Fref{fig1} shows $x_c(c)$ as a function of average degree $c$.
By evaluating the linear instability, it is apparent that $c^{\mathrm{LP}}=c^{\mathrm{IP}}\simeq 5.239$.
Below the threshold, RS solutions in both limits merge, but they split otherwise.
The RS solution in the LP-limit rapidly approaches to the theoretical upper bound.
We also perform numerical simulations of LP relaxation using 
lp\_solve solver~\cite{lpsolve} with the revised simplex method.
They agree well with the analytical estimations suggesting the correctness of statistical--mechanical analyses.
These facts show that there exists a graph ensemble in which typical
behavior of LR is the worst among three approximation methods.
We present the average half-integral ratio $p_h(c)$ in \fref{fig2}.
Although finite-size effects are observed for numerical results, we find the analytical result for $p_h(c)$ is also asymptotically correct.
In other words, LP possibly finds some fraction of integer variables even if the whole graph is an LR core.

\begin{figure}
\begin{center}
 \includegraphics[trim=0 0 0 0,width=0.8\linewidth]{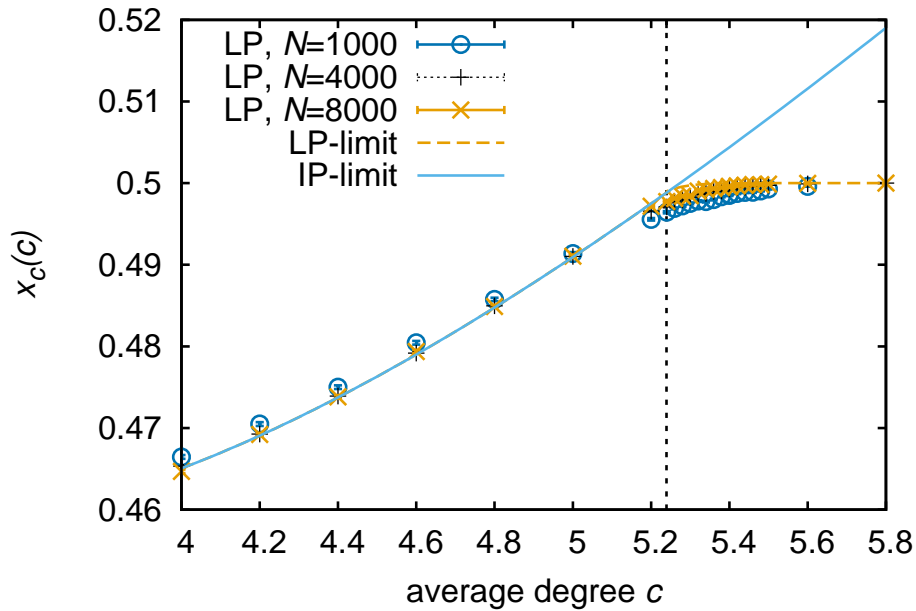}
 \caption{(Color online) Average optimal value $x_c(c)$ for the
 min-VC on BA-like scale free networks as a function of the average degree $c$.
 Solid and dashed lines respectively represent RS solutions in the IP-limit and LP-limit.
 The vertical dashed line corresponds to the threshold $c^{\mathrm{LP}}=c^{\mathrm{IP}}\simeq 5.239$, above which
 the linear stability in the IP-limit breaks.
 Symbols are numerical results of LP relaxation using lp\_solve solver.
 They are averaged, respectively, over $1600$, $1000$, and $400$ graphs with cardinality $N=1000$, $4000$, and $8000$. 
}
 \label{fig1}
\end{center}
\end{figure}

\begin{figure}
\begin{center}
 \includegraphics[trim=0 0 0 0,width=0.8\linewidth]{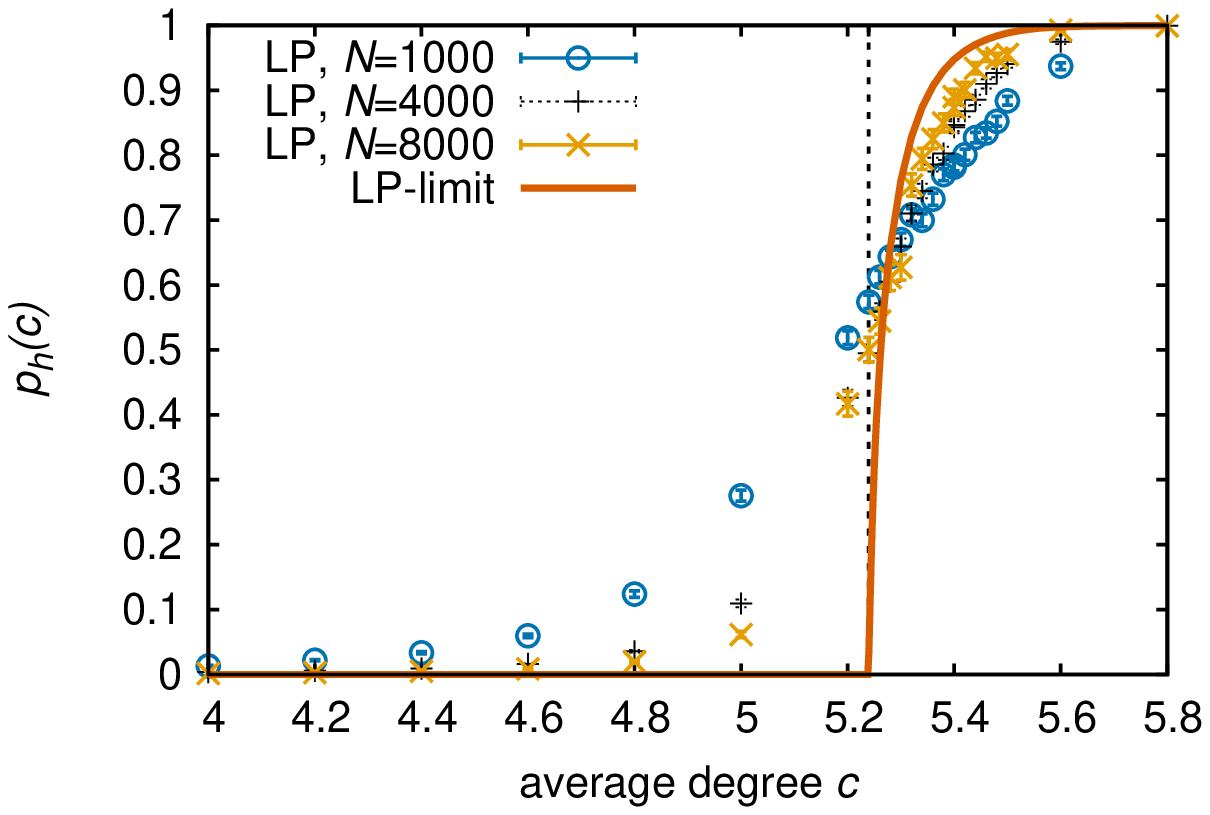}
 \caption{(Color online) Average half-integral ratio $p_h(c)$ for the
 min-VC on BA-like scale free networks as a function of the average degree $c$.
 The solid line represents RS estimations of the $p_h(c)$ in the LP-limit.
 The vertical dashed line is $c^{\mathrm{LP}}$, above which numerous half-integers emerge.
 Symbols are numerical results of LP relaxation using lp\_solve solver.
 They are averaged, respectively, over $1600$, $1000$, and $400$ graphs with cardinality $N=1000$, $4000$, and $8000$. 
}
 \label{fig2}
\end{center}
\end{figure}

\subsection{Scale-free networks with continuous power}
Here we provide a more general class of random graphs in which the power of the degree distribution can be tuned.
The degree distribution is represented as
\begin{equation}
p_k=\cases{
C_{0}m^{-\gamma} &($k=m$),\\
(1-p)C_{0}k^{-\gamma}+pC_{1}k^{-\gamma} &($k>m$), 
}
\label{eq_c1}
\end{equation}
where $0\le p<1$, $\gamma>2$, and $m\in \mathbb{N}$.
$C_{a}^{-1}$ ($a=0,1$) is a normalization factor given by the generalized Riemann zeta function $\zeta(\gamma,m+a)\equiv\sum_{k=m+a}^\infty k^{-\gamma}$.
The degree is bounded by $m$. Its average is given as a function of $m$, $p$, and $\gamma$.
It reads
\begin{equation}
c(m,p,\gamma)=\frac{\zeta(\gamma-1,m)}{\zeta(\gamma,m)}(1-p)+\frac{\zeta(\gamma-1,m+1)}{\zeta(\gamma,m+1)}p.\label{eq_c2}
\end{equation}

Fixing $\gamma$, the threshold of the RS--RSB transition is given as $m^{\mathrm{IP}}(\gamma)$ and $p^{\mathrm{IP}}(\gamma)$.
The critical average degree is therefore given as $c^{\mathrm{IP}}(\gamma)=c(m^{\mathrm{IP}}(\gamma),p^{\mathrm{IP}}(\gamma),\gamma)$.
We show $c^{\mathrm{IP}}(\gamma)$ and $m^{\mathrm{IP}}(\gamma)$ in \fref{fig3}.
Because $\gamma$ is close to $2$, both diverge.
However, $\lim_{\gamma\rightarrow\infty}c^{\mathrm{IP}}(\gamma)=2$ because the graph ensemble converges to $2$-regular random graphs.
\Fref{fig4} shows the average cover ratio $x_c^{\mathrm{IP}}$ at the critical average degree $c^{\mathrm{IP}}(\gamma)$.
There exists some $\gamma$ at which $x_c^{\mathrm{IP}}(c^{\mathrm{IP}})>1/2$.
Considering the upper bound of LP relaxation, it fails good estimations at $c^{\mathrm{LP}}$, where $x_c^{\mathrm{LP}}(c^{\mathrm{LP}})=1/2$.
This fact suggests that there exists a third case for which LP relaxation has no good approximations.
 Nevertheless loopy BP still works well.

\begin{figure}
\begin{center}
 \includegraphics[trim=0 0 0 0,width=0.8\linewidth]{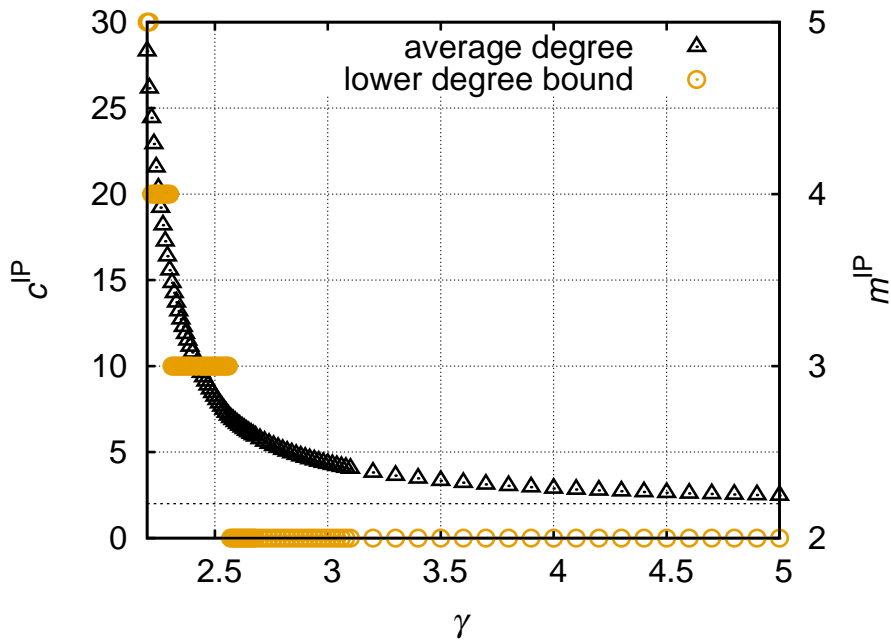}
 \caption{(Color online) Critical average degree $c^{\mathrm{IP}}$ (triangle) and critical lower degree bound $m^{\mathrm{IP}}$ (circle)
  as a function of the power $\gamma$. The horizontal dotted line represents $c^{\mathrm{IP}}=2$, 
  limiting value of $c^{\mathrm{IP}}$ as $\gamma\rightarrow \infty$.
}
 \label{fig3}
\end{center}
\end{figure}

\begin{figure}
\begin{center}
 \includegraphics[trim=0 0 0 0,width=0.8\linewidth]{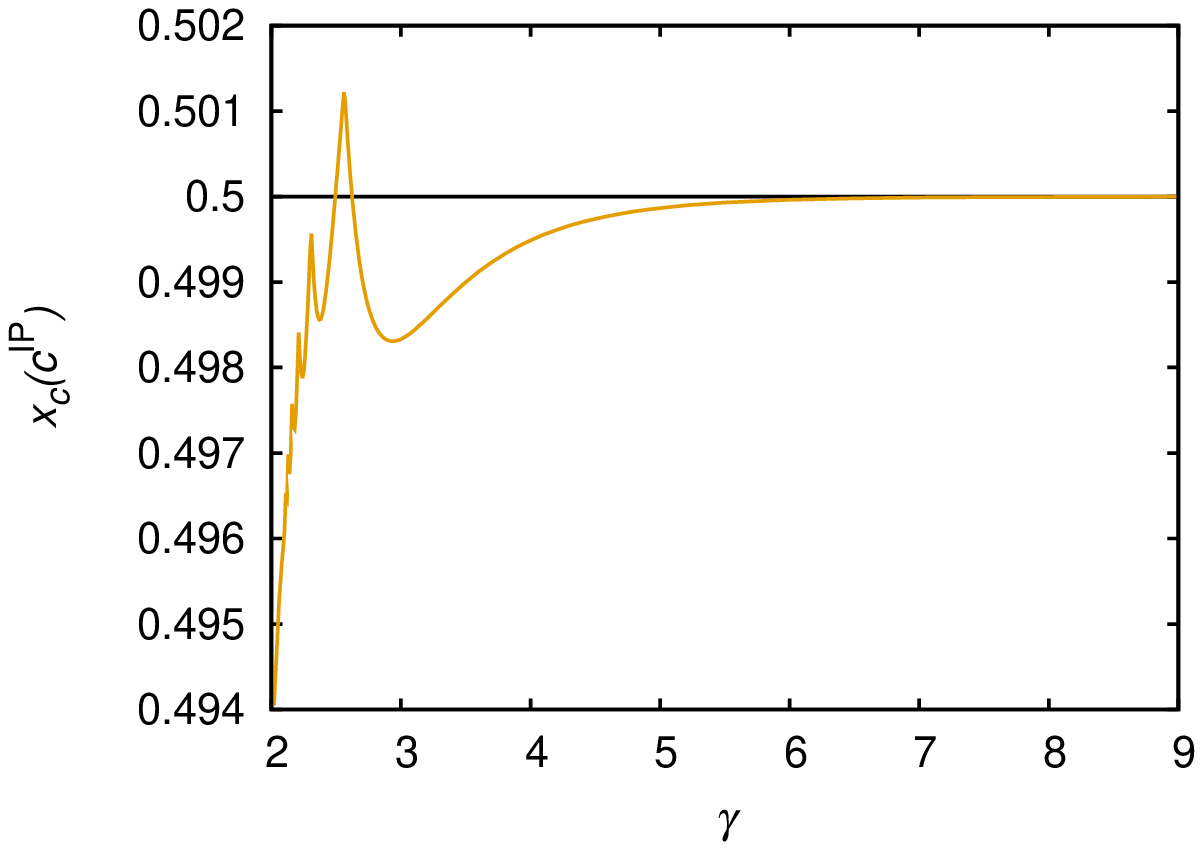}
\caption{(Color online) Average optimal value $x_c^{\mathrm{IP}}(c^{\mathrm{IP}})$ at which the RS solution is unstable
 as a function of the power $\gamma$.
 The horizontal line represents $x_c=1/2$, the upper bound of LP relaxation.
}
 \label{fig4}
\end{center}
\end{figure}

Here, we present an example of the graph ensemble showing $c^{\mathrm{LR}}<c^{\mathrm{LP}}<c^{\mathrm{IP}}$.
\Fref{fig5} shows $x_c(c)$ with $\gamma=2.56$.
The linear stability of the RS solution in the IP-limit breaks at $c^{\mathrm{IP}}\simeq 7.133$.
$x_c^{\mathrm{IP}}$ is not smooth near $c=7.1$ because the lower degree bound increases.
It apparently exceeds a half suggesting the failure of LP relaxation.
The RS solution merges to that in the IP-limit but splits at $c^{\mathrm{LP}}\simeq 7.07$.
Equation \eref{eq_m6} has two stable
solutions of $(X,Y)=(1,0)$ and $(x,x)$ ($0<x<1$) if $c^{\mathrm{LP}}<c<c^{\mathrm{IP}}$.
The solution $(X,Y)=(1,0)$ gives the lower free-energy resulting in $x_c^{\mathrm{LP}}(c)=1/2$.
Consequently, $x_c^{\mathrm{LP}}(c)$ is bent at $c=c^{\mathrm{LP}}<c^{\mathrm{IP}}$, which is an example of case (iii) in \sref{sec_dif}.

\begin{figure}
\begin{center}
 \includegraphics[trim=0 0 0 0,width=0.8\linewidth]{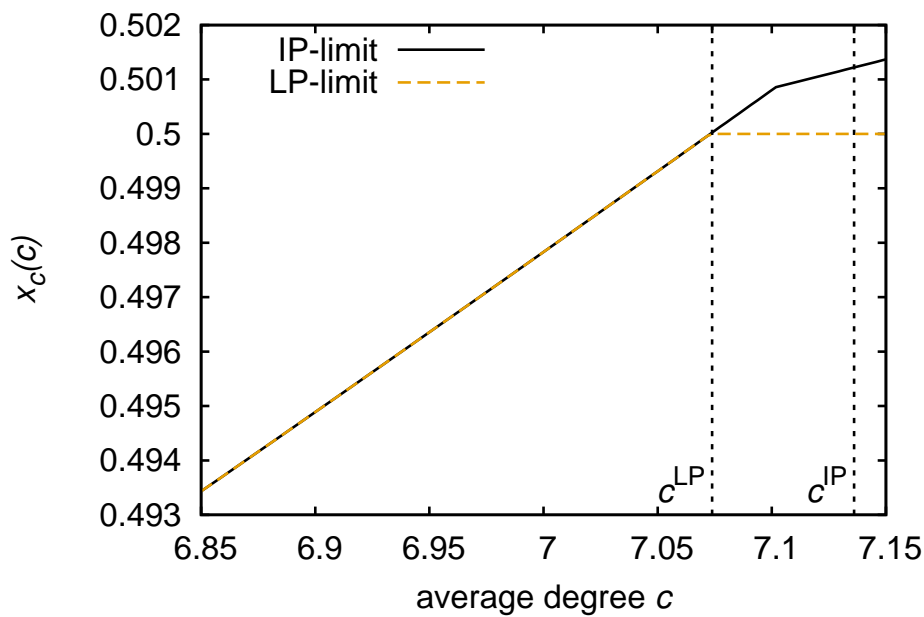}
 \caption{(Color online) Average optimal value $x_c(c)$ as a function of the average degree $c$ in a case in which $\gamma=2.56$.
 Solid and dashed lines respectively represent the RS solutions in the IP-limit and LP-limit.
 The vertical dashed lines shows $c^{\mathrm{LP}}$ and $c^{\mathrm{IP}}$. 
}
 \label{fig5}
\end{center}
\end{figure}

Thereby, we obtain examples for all cases discussed in \sref{sec_dif}.
By introducing random scale-free networks with continuous average degree and the lower degree bound,
  it is apparent that case (ii) emerges even in inhomogeneous network ensembles.
In this case, LP relaxation and the loopy BP work better than LR on random graphs with a finite range of average degree.
We also demonstrate an example of case (iii).
In this case, the loopy BP possibly works well even if LR and LP cannot work because of their characteristics.
Our example also shows that the power-law behavior of degree distribution is necessary for case (iii).

\section{Summary and discussions}

As described in this paper, we evaluate the typical behavior of approximation algorithms for min-VC using some theoretical analyses.
Instead of the conventional homogeneous random graphs, the typical performance over a graph ensemble with an arbitrary degree distribution
is studied. 
Convex optimization theory reveals that LP always solves the original min-VCs exactly if LR finds true optimal solutions.
We also use the generating function method for LR and the RS
cavity method for LP and BP to estimate the typical-performance thresholds.
As a result, we clarify that, in some cases, three algorithms have different thresholds above which they fail to approximate
the problem with good accuracy.
If the fraction of vertices with degree one is not zero, i.e., LR can work, they have the same threshold
 for the RS--RSB phase transition.
Otherwise, LR cannot work at all, giving it the worst performance among the three approximations.
It is widely observed for random graphs with lower degree bound greater than one.
The phase in which only LR cannot approximate the problem typically has a finite region if the degree distribution follows a power law.
As the last case, we provide an example for which their thresholds
apparently differ from each other.
This unusual case occurs for the min-VCs on random graphs satisfying two conditions: $p_1=0$ and $x_c^{\mathrm{IP}}(c^{\mathrm{IP}})> 1/2$.

Considering that the min-VC is defined on a graph, hard problems have some graph structures that make problems difficult.
The LR core and LP core (parts of graphs with half-integral LP solutions) are its candidates.
It is denied, however, because of the existence of case (iii) in \sref{sec_dif}.
Moreover, in our power-law degree distribution model, the appropriate
parameter $\gamma$ for case (iii) is very limited, which indicates that the graph structures for which only BP works are
well affected not only by a scale-free property but also
by other conditions such as their lower degree bound.
It is a difficult but meaningful task for future work to characterize a structure for hard problems by graph invariants.
In the sense of the graph structure, graph properties such as a degree--degree correlation and clustering structure
 neglected in this paper must be examined.
These properties will be properly reflected in statistical--mechanical analyses using generalized cavity methods proposed
 in~\cite{vazquez2003computational,decelle2011asymptotic}.

Our theoretical analyses claim that LP and BP
 have the same threshold of typical performance 
 in the wide range of random graphs.
However, the threshold $c^{\mathrm{IP}}$ calculated using the RS ansatz is a necessary condition 
for a typical-performance threshold $c^{\mathrm{BP}}$ of BP because it fails to converge above the dynamical transition 
of the one-step replica symmetry breaking (1RSB) if it exists.
Although this dynamical phase exists in some randomized constraint satisfaction problems~\cite{Mezard2003},
it has not been discovered in ground states of the randomized min-VCs.
To investigate its existence, one must write down a functional equation based on the 1RSB ansatz~\cite{Zhang2009},
which is difficult to solve analytically for arbitrary random graphs unfortunately.
The LP or LR possibly show better typical performance than BP if the dynamical phase exists for some ensembles.
It is an interesting subject for future work to investigate the existence of the transition and typical performance of approximations for a certain ensemble.
In contrast, the equivalence of the typical performance of BP and LP 
 is shown mathematically using probabilistic analysis in a specified case~\cite{Sanghavi2005}.
In this sense, our analyses provide a general conjecture on the typical behavior of approximation algorithms.
It is interesting to extend probabilistic analysis of LP relaxation for the min-VC to a more general case.

As described in this paper, we combine several approaches to average-case analyses for approximation algorithms.
These analyses are based on their algorithmic properties.
Especially, the average-case analysis of LP relaxation (including the probabilistic analysis above)
 is based on half-integrality of LP-relaxed min-VCs.
Some numerical results suggest, however, that a relation of the
typical performance between LP and BP is expected for more general situations
 without the half-integrality property~\cite{takabe2016statistical}.
Establishing their general connection is important from the viewpoint of continuous relaxation for discrete optimizations.
Along with LP relaxation, the semidefinite programming relaxation for the strict quadratic programming problems is analyzed
 using the RS cavity method~\cite{javanmard2015phase}.
Statistical--mechanical analyses will be helpful to extend a probabilistic analysis such as~\cite{coja2006max}.
We hope that this paper stimulates further studies of the typical behavior of approximation algorithms and its connection to the spin-glass theory.

\ack
  ST thanks to T. Hasegawa for fruitful discussions related to complex networks.
 This research was supported by Grants-in-Aid for Scientific Research
 from MEXT, 
  Japan (Nos. 22340109, 25610102, and 25120010), and for JSPS fellows (No. 15J09001).

\appendix
\section{RS cavity analysis of the min-VC and its LP relaxation}
 Here we describe the mean-field analysis for the LP--IP model in
 detail. 
 For simplicity of representation, a spin variable $\sigma_i=1-x_i$ is used instead of $x_i$ itself.
 The LP--IP model on graph $G=(V,E)$ is described by the Hamiltonian
\begin{equation}
\mathcal{H}_r(\bm{\sigma})=-\sum_{i=1}^N \sigma_i+\mu^{r-1}\sum_{i}\delta_{\sigma_i,1/2}, \label{eq_app1}
\end{equation}
 where $\bm{\sigma}=\{\sigma_i\}=\{0,1/2,1\}^N$ and $\sigma_i=1$ corresponds to uncovered state of variable node $i$.
The grand canonical partition function is then defined as
\begin{equation}
\Xi=\sum_{\bm{\sigma}}\exp(-\mu\mathcal{H}_r(\bm{\sigma}))\prod_{(i,j)\in E}\theta(1-\sigma_i-\sigma_j). \label{eq_m2}
\end{equation}

 First, we obtain BP equations as described in \sref{sec_def}.
By Bethe--Peierls approximation, the probability of $\sigma_i=\sigma$ is
\begin{equation}
P_i(\sigma)\simeq \frac{1}{Z_i}\exp(\mu\sigma-\mu^{r}\delta_{\sigma,1/2})\prod_{j\in\partial i}P_{j\rightarrow i}(\sigma_j)
,\label{eq_ap1}
\end{equation}
 where $P_{j\rightarrow i}(\sigma)$ is the probability of $\sigma_j=\sigma$ on a cavity graph $G\backslash i$.
 The probability $P_{j\rightarrow i}(\sigma_j)$ is regarded as a message on
 the graph and satisfies the following recursive relation as
\begin{equation}
P_{i\rightarrow j}(\sigma_i)\simeq \frac{1}{Z_{i\rightarrow j}}\exp(\mu \sigma_i-\mu^r \delta_{\sigma_i,1/2})
\prod_{k\in\partial i\backslash j}\sum_{\sigma_k}P_{k\rightarrow i}(\sigma_k)\theta(1-\sigma_i-\sigma_k). \label{eq_ap2}
\end{equation}
 By substituting a spin value, we obtain
\begin{eqnarray}
P_{i\rightarrow j}(1)&\simeq \frac{1}{Z_{i\rightarrow j}}e^\mu
\prod_{k\in\partial i\backslash j}\left(1-P_{k\rightarrow i}(1)-P_{k\rightarrow i}\left(\frac{1}{2}\right)\right), \nonumber\\
P_{i\rightarrow j}\left(\frac{1}{2}\right)&\simeq \frac{1}{Z_{i\rightarrow j}}\exp\left(\frac{\mu}{2}-\mu^r\right)
\prod_{k\in\partial i\backslash j}\left(1-P_{k\rightarrow i}(1)\right), \nonumber\\
P_{i\rightarrow j}(0)&\simeq \frac{1}{Z_{i\rightarrow j}}.  \label{eq_ap4}
\end{eqnarray}
 
 As we take the $\mu\rightarrow \infty$ limit, we rescale the messages
 by introducing effective fields $\nu_{i\rightarrow a}$ and
$\xi_{i\rightarrow a}$ defined as
\begin{equation}
P_{i\rightarrow j}(1)\equiv \frac{e^{\mu\xi_{i\rightarrow a}}}{1+e^{\mu\xi_{i\rightarrow a}}+e^{\mu\nu_{i\rightarrow a}}},\quad
P_{i\rightarrow j}\left(\frac{1}{2}\right)\equiv \frac{e^{\mu\nu_{i\rightarrow a}}}{1+e^{\mu\xi_{i\rightarrow a}}+e^{\mu\nu_{i\rightarrow a}}}.
  \label{eq_ap4b}
\end{equation}
 Eq.~(\ref{eq_ap4b}) enables us to write down BP equations for these
 fields as 
\begin{eqnarray}
\xi_{i\rightarrow j}&=1-\frac{1}{\mu}\sum_{k\in\partial i\backslash j}
\ln\left[1+e^{\mu\xi_{k\rightarrow i}}+e^{\mu\nu_{k\rightarrow i}}\right],\nonumber\\
\nu_{i\rightarrow j}&=\frac{1}{2}-\mu^{r-1}+\frac{1}{\mu}\sum_{k\in\partial i\backslash j}
\ln\left[\frac{1+e^{\mu\nu_{k\rightarrow i}}}{1+e^{\mu\xi_{k\rightarrow i}}+e^{\mu\nu_{k\rightarrow i}}}\right]
.  \label{eq_ap5}
\end{eqnarray}

 We then consider a graph ensemble for which the degree distribution of variable nodes is $p_k$ ($k\ge 0$).
 Let $P(\xi,\nu)$ be the frequency distribution of a set of fields $(\xi,\nu)$.
 From Eq.~(\ref{eq_ap5}), we find a self-consistent equation of ${P}$ as
\begin{eqnarray}
{P}(\xi,\nu)&=\sum_{k=0}^\infty\frac{kp_k}{c}\int\prod_{i=1}^{k-1}d{P}
\left(\xi^{(i)},\nu^{(i)}\right)\nonumber\\
&\times\delta\left(\xi-1+\frac{1}{\mu}\sum_{i}\ln\left[1+e^{\mu\xi^{(i)}}+e^{\mu\nu^{(i)}}\right]\right)\nonumber\\
&\times\delta\left(\nu-\frac{1}{2}+\mu^{r-1}-\frac{1}{\mu}\sum_{i}
\ln\left[\frac{1+e^{\mu\nu^{(i)}}}{1+e^{\mu\xi^{(i)}}+e^{\mu\nu^{(i)}}}\right]\right)
.
  \label{eq_ap6}
\end{eqnarray}

 The first limit is the IP--limit with $r>1$.
 In this limit, the cavity field $\nu$ negatively diverges as $\mu\rightarrow \infty$, which corresponds to the fact that its ground states consist of no half-integral spin values.
 Let $X$ be the probability that $\xi$ is positive in this limit.
 From \eref{eq_ap6}, we obtain an equation of $X$, which reads
\begin{equation}
 X=\sum_k\frac{kp_k}{c}(1-X)^{k-1}. \label{eq_ap7}
\end{equation}
Using the solution, the minimum cover ratio is then represented as
\begin{equation}
 x_c^{\mathrm{IP}}(c)=1-\frac{1}{2}\sum_kp_k\left\{2(1-X)^{k}+kX(1-X)^{k-1}\right\}. \label{eq_ap8}
\end{equation} 
  
 Next, we specifically examine the LP-limit with $0<r<1$.
 As described in~\cite{takabe2016statistical}, numerical solutions of \eref{eq_ap6} have a support around some lattice points.
 We therefore apply a discretized ansatz for cavity fields:
 weights around $(\xi,\nu)=(1,1/2)$ and $(1/2,1/2)$ are defined respectively as $r_1$ and $r_2$.
 We also set the marginal probability for which $\xi\le 0$ and $\nu=1/2$ to $r_3$.

 By substituting $X=r_1+r_2+r_3,Y=r_1(\le X)$, finally we obtain the
 recursive relations as  
\begin{equation}
 X=\sum_k\frac{kp_k}{c}(1-Y)^{k-1},\quad Y=\sum_k\frac{kp_k}{c}(1-X)^{k-1}. \label{eq_ap9}
\end{equation}
 This is a recursive relation denoted in \sref{sec_lpip}.

 Considering a small penalty $\mu^{r-1}$, it is straightforward to
 obtain a marginal distribution of each spin via \eref{eq_ap1}, which reads
\begin{eqnarray}
\fl \mbox{Pr}(\sigma_i=1)&=\sum_{k}p_k\left\{(1-X)^k+k(X-Y)(1-X)^{k-1}\right\},\nonumber\\
\fl \mbox{Pr}\left(\sigma_i=\frac{1}{2}\right)&
=\sum_{k}p_k\left\{(1-Y)^k-(1-X)^k-k(X-Y)(1-X)^{k-1}\right\},\nonumber\\
\fl \mbox{Pr}\left(\sigma_i=1\mbox{ w.p. } \frac{1}{2} \mbox{ and }0 \mbox{ w.p. }\frac{1}{2}\right)&=\sum_{k}kp_kY(1-X)^{k-1}.
  \label{eq_ap10}
\end{eqnarray}
 These lead to the average minimum cover ratio
\begin{equation}
 x_c^{\mathrm{LP}}(c)=1-\frac{1}{2}\sum_kp_k\left\{(1-Y)^k+(1-X)^{k}+kX(1-X)^{k-1}\right\}, \label{eq_ap11}
\end{equation}
 and the average ratio of half-integral spins
\begin{eqnarray}
 p_h(c)&=\sum_{k}p_k\left\{(1-Y)^k-(1-X)^k-k(X-Y)(1-X)^{k-1}\right\}. \label{eq_ap12}
\end{eqnarray}

\section*{References}
\bibliographystyle{plain}
\bibliography{VCapp}

\end{document}